\documentclass{article}

\usepackage[english]{babel}

\usepackage[letterpaper,top=2cm,bottom=2cm,left=3cm,right=3cm,marginparwidth=1.75cm]{geometry}

\usepackage[T1]{fontenc}
\usepackage{amsfonts,amsmath,amsthm,amssymb,mathtools}  %

\mathtoolsset{centercolon}
\usepackage{xfrac,nicefrac}
\usepackage{mathdots}
\usepackage{mleftright}  %

\usepackage{xspace}
\xspaceaddexceptions{]\}}  %
\usepackage{regexpatch}

\usepackage{bm,bbm,dsfont}  %
\usepackage{caption}
\usepackage[normalem]{ulem}
\usepackage{enumitem}

\usepackage{graphicx}
\usepackage{float}
\usepackage{subcaption}  %
\usepackage{tcolorbox}
\usepackage{tikz}
\usetikzlibrary{decorations.pathreplacing}
\usetikzlibrary{calc}
\usetikzlibrary{positioning}
\usetikzlibrary{arrows.meta}

\usepackage[linesnumbered,boxed,ruled,vlined]{algorithm2e}
\usepackage{algpseudocode}

\usepackage{thmtools,thm-restate}
\theoremstyle{plain}

\newtheorem{theorem}{Theorem}[section]  %
\newtheorem{lemma}[theorem]{Lemma}

\theoremstyle{definition}  %

\usepackage[colorlinks,citecolor=blue,linkcolor=blue,urlcolor=red]{hyperref}
\usepackage[capitalise]{cleveref}
\crefname{algocf}{Algorithm}{Algorithms}
\Crefname{algocf}{Algorithm}{Algorithms}

\DeclarePairedDelimiter{\bk}{(}{)}

\DeclarePairedDelimiter{\BK}{\{}{\}}

\DeclarePairedDelimiter{\abs}{\lvert}{\rvert}

\DeclareMathOperator{\poly}{poly}

\renewcommand{\tilde}{\widetilde}

\newcommand{\defeq}{\coloneqq}
\newcommand{\eps}{\varepsilon}

\renewcommand{\emptyset}{\varnothing}
\renewcommand{\epsilon}{\eps}

\usepackage{adjustbox}
\usepackage[numbers]{natbib}
\definecolor{darkgreen}{RGB}{0,128,0}

\title{Listing 6-Cycles}
\author{Ce Jin\thanks{\url{cejin@mit.edu}, Partially supported by NSF Grant CCF-2129139.} \\ MIT \and Virginia Vassilevska Williams\thanks{\url{virgi@mit.edu}, Supported by NSF Grants CCF-2129139 and CCF-2330048 and BSF Grant 2020356.} \\ MIT \and Renfei Zhou\thanks{\url{zhourf20@mails.tsinghua.edu.cn}} \\ Tsinghua University}
\date{}

\begin{document}
\maketitle

\begin{abstract}
Listing copies of small subgraphs (such as triangles, $4$-cycles, small cliques) in the input graph is an important and well-studied problem in algorithmic graph theory.
In this paper, we give a simple algorithm that lists $t$ (non-induced) $6$-cycles in an $n$-node undirected graph in $\tilde O(n^2+t)$ time.
This nearly matches the fastest known algorithm for detecting a $6$-cycle in $O(n^2)$ time by Yuster and Zwick (1997).
Previously, a folklore $O(n^2+t)$-time algorithm was known for the task of listing $4$-cycles.
\end{abstract}

\section{Introduction}
Listing (also called enumerating) cycles in a  graph is an important algorithmic task with 
a variety of applications  (e.g. in computational biology \cite{bio} and social networks \cite{GiscardRW17}). A line of works dating back to the 1970s developed efficient algorithms for listing all simple cycles in a (directed or undirected) input graph \cite{Tarjan73,Johnson75,ReadT75,MatetiD76,bezem1987enumeration,BirmeleFGMPRS13,Grossi16}, and some of these algorithms have been engineered in practice \cite{BlanusaIA22}. 
These algorithms consider cycles of arbitrarily large length; however, in many natural scenarios, one only cares about listing cycles whose length is a small given constant, and this will be the focus of this paper.

Formally, we are given an \emph{undirected} simple graph $G$ of $n$ nodes and $m$ edges,  and a small pattern graph $H$ (in our case, $H$ is the $k$-cycle $C_k$ for some small fixed $k$). 
The $H$-\emph{detection} problem asks whether $G$ contains a copy of $H$ as a (\emph{not necessarily induced}) subgraph.
In the $H$-\emph{listing} problem, we are additionally given a parameter $t$, and we need to output $t$ distinct copies of $H$ in $G$ (if $G$ contains fewer than $t$ copies of $H$, then we need to output all copies).

\paragraph*{Known results on $C_k$-detection and $C_k$-listing.}
The simplest cycle is a triangle ($C_3$). A long line of works in graph algorithms and fine-grained complexity has led to a complete picture of triangle listing algorithms, up to natural hardness hypotheses: 
\cite{BjorklundPWZ14} gave an algorithm for listing $t$ triangles in $\tilde O(\min\{n^2 + nt^{2/3}, m^{4/3}+mt^{1/3}\})$ time\footnote{We use $\tilde O(f)$ to denote $O(f\cdot \mathrm{polylog}(f))$.},
assuming that the matrix multiplication exponent is $\omega = 2$.
This running time is likely to be optimal:  for $t=1$ it matches the best known triangle detection algorithm in $O(m^{2\omega/(\omega+1)})$ time \cite{AlonYZ97}, and when $t$ is sufficiently large in terms of $n$ (or $m$) this running time is conditionally optimal under the 3SUM hypothesis \cite{Patrascu10,KopelowitzPP16} and even more believable hypotheses \cite{WilliamsX20}. See also \cite{DurajK0W20} for more fine-grained reductions between triangle listing and other problems.

Now we review known results for cycle length $k>3$. A common phenomenon is that finding {\em even} length cycles appears to be easier than finding odd cycles.
Finding a cycle of any fixed odd length in undirected graphs is as difficult as the problem in directed graphs (see e.g. \cite{thesis}), whereas finding even cycles is easier in undirected graphs than in directed graphs. Bondy and Simonovits \cite{bondy} showed that for any integer $k\geq 2$, any $n$ vertex graph with at least $100k n^{1+1/k}$ edges must contain a $2k$-cycle. Meanwhile, the complete bipartite graph on $n$ nodes in each partition does not contain any cycles of odd length.
This difference has been exploited to obtain faster algorithms for finding a cycle of any given even length \cite{Yuster97FindEven,DKS17CapWalk} that are faster than the known algorithms for detecting even the smallest length odd cycle, a triangle.

Let us review the previous works on the easier cycle \emph{detection} problem in more detail.
Yuster and Zwick \cite{Yuster97FindEven} gave an $O(n^2)$-time algorithm that detects  $C_{k}$ in an undirected $n$-node graph for any fixed \emph{even} integer $k$. Detecting $C_k$ when $k=O(1)$ is odd can be done in $\tilde{O}(n^\omega)$ time (e.g. via color-coding \cite{Alon95ColorCoding}).

When graph sparsity is taken into account, and the running time is in the number of edges $m$, the fastest running times are as follows.
The fastest algorithms for $C_k$-detection for  odd $k$ in terms of $m$ run in $O(m^{\omega(k+1)/(2\omega+k-1)})$ time \cite{DalirrooyfardVW21} and $O(m^{2-\frac{2}{k+1}})$ time \cite{AlonYZ97}, where the former is faster than the latter if $\omega\leq 2+2/(k-1)$; for the current best value of $\omega$ \cite{matmult-duan,matmult-new} this is the case for $k\leq 5$. 
For even $k\geq 4$, the best known running time for detecting $C_k$ is $O(m^{2-\frac{4}{k+2}})$ \cite{DKS17CapWalk,AlonYZ97}; this running time was shown to be conditionally optimal for $k\equiv 2\pmod{4}$ for algorithms that do not use fast matrix multiplication \cite{Lincoln020,DKS17CapWalk}.

While the listing variant of the problem is well understood for triangles ($C_3$), there isn't very much known for larger $k$.
The problem of $C_4$-listing was very recently independently studied by \cite{fourcycle} and \cite{JinX23}, who gave an algorithm in $\tilde O(\min\{n^2+t,m^{4/3}+t\})$ time. For $t=1$ this matches the known $C_4$-detection algorithms in $O(n^2)$ time \cite{Yuster97FindEven} and $O(m^{4/3})$ time \cite{AlonYZ97}. It was also shown to be conditionally optimal under 3SUM hypothesis independently by \cite{AbboudBF23} and \cite{JinX23} (building on \cite{AbboudBKZ22}), and later shown to be conditionally optimal under the Exact Triangle hypothesis (which is weaker than 3SUM hypothesis and APSP hypothesis) \cite{chanxu}.

\paragraph*{Our results.}
Given the previous works on the fine-grained complexity of listing triangles and $4$-cycles, a natural question is to investigate the fine-grained complexity of $C_k$-listing for larger $k$.
In this paper, we make the first step by considering $k=6$, the smallest even number after $4$. Our result is summarized in the following theorem.
\begin{theorem}[List $t$ 6-cycles]
\label{thm:mainlist}
    There is a deterministic algorithm that lists $t$ copies of $C_6$ in an $n$-node undirected simple graph in $\tilde O(n^2+t)$ time.
\end{theorem}
The running time of our algorithm nearly matches the algorithm of Yuster and Zwick \cite{Yuster97FindEven} for detecting a $C_6$ in $O(n^2)$ time, for any $t\leq O(n^2)$.

\newcommand{\PathForfourcycle}{\adjustbox{valign=m, raise=-0.6mm, margin=-1mm 0 0 0}{
    \begin{tikzpicture}[scale=0.5,every node/.style={inner sep=1pt}]
      \node (b) at (1,0) {$a$};
      \node (d) at (3.4,0) {$c$};
      \node (c1) at (2.2,0.5) {$b_1$};
      \node (c2) at (2.2,-0.5) {$b_2$};
      \draw[-] (b) -- (c1) -- (d);
      \draw[-] (b) -- (c2) -- (d);
    \end{tikzpicture}}}
\newcommand{\TwoPaths}{\adjustbox{valign=m, raise=-0.6mm, margin=-1mm 0 0 0}{
    \begin{tikzpicture}[scale=0.5,every node/.style={inner sep=1pt}]
      \node (a) at (0,0) {$a$};
      \node (d) at (3.6,0) {$d$};
      \node (b1) at (1.2,0.5) {$b_1$};
      \node (b2) at (1.2,-0.5) {$b_2$};
      \node (c1) at (2.4,0.5) {$c_1$};
      \node (c2) at (2.4,-0.5) {$c_2$};
      \draw[-] (a) -- (b1) -- (c1) -- (d);
      \draw[-] (a) -- (b2) -- (c2) -- (d);
    \end{tikzpicture}}}

At a high level, our algorithm for listing 6-cycles follows a similar idea as in the folklore 4-cycle detection/listing algorithm (e.g., \cite{Yuster97FindEven,JinX23,fourcycle}), where a $4$-cycle \PathForfourcycle$\,$ is formed by pasting together two paths $a - b_1 - c$ and $a - b_2 - c$ (where $b_1\neq b_2$) which are stored in a table indexed by $(a,c)$. Analogously, our algorithm forms a 6-cycle \TwoPaths$\,$ by pasting together two paths $a-b_1-c_1-d$ and $a-b_2-c_2-d$, but the new challenge here is to avoid wasting time on  degenerate cases $b_1=b_2$ or $c_1=c_2$ which do not produce valid $6$-cycles.

We leave it as an open question to extend \cref{thm:mainlist} from $C_6$ to $C_{2k}$ for larger $k$.

\paragraph*{Further related works.}
The problem of listing short cycles was also studied in restricted graph classes such as low-arboricity graphs \cite{ChibaN85} and planar graphs \cite{Kowalik03}.

Another way to generalize the classical triangle listing results is to consider $k$-cliques of larger size $k$. A very recent work \cite{listclique} considered the problem of listing $k$-cliques for small $k$.

More generally, enumeration algorithms are widely studied in database theory, and recently there is growing interest in the fine-grained complexity of enumeration algorithms; see e.g., \cite{DurandG07, Strozecki19,durandtutorial,CarmeliK21,BringmannCM22,bringmanncarmeli,DengL023,CarmeliS23}.

In the literature of enumeration algorithms, it is common to study the \emph{delay} between outputting two consecutive answers.
For example, the 4-cycle listing algorithm of \cite{JinX23} actually achieves $O(1)$ delay after $O(\min\{n^{2},m^{4/3}\})$-time preprocessing; this is  stronger than listing $t$ $4$-cycles in $O(\min\{n^{2},m^{4/3}\} + t)$ time. Whether our $6$-cycle listing algorithm (\cref{thm:mainlist}) can be upgraded to an enumeration algorithm with $\poly \log(n)$-delay after $\tilde O(n^2)$-time preprocessing is left for future investigation.

\section{Algorithm for Listing All 6-Cycles}

\newcommand{\PathForP}{\adjustbox{valign=m, raise=-0.6mm, margin=-1mm 0 0 0}{
    \begin{tikzpicture}[scale=0.5,every node/.style={inner sep=1pt}]
      \node (a) at (0,0) {$a$};
      \node (b) at (1,0) {$b$};
      \node (d) at (3.4,0) {$d$};
      \node (c1) at (2.2,0.5) {$c_1$};
      \node (c2) at (2.2,-0.5) {$c_2$};
      \draw[-] (a) -- (b) -- (c1) -- (d);
      \draw[-] (b) -- (c2) -- (d);
    \end{tikzpicture}}}

\newcommand{\PathForQ}{\adjustbox{valign=m, raise=-0.6mm, margin=-1mm 0 0 0}{
    \begin{tikzpicture}[scale=0.5,every node/.style={inner sep=1pt}]
      \node (a) at (0,0) {$a$};
      \node (d) at (3.4,0) {$d$};
      \node (b1) at (1.2,0.5) {$b_1$};
      \node (b2) at (1.2,-0.5) {$b_2$};
      \node (c) at (2.4,0) {$c$};
      \draw[-] (a) -- (b1) -- (c) -- (d);
      \draw[-] (a) -- (b2) -- (c);
    \end{tikzpicture}}}

Let $G = (V, E)$ be an undirected graph with $n$ nodes.
We first focus on a slightly easier variant of the 6-cycle listing problem which asks to list \emph{all} 6-cycles in $G$. Our algorithm for this variant is summarized in the following theorem.
\begin{theorem}[List all 6-cycles]
  \label{thm:all}
  There is a deterministic algorithm that lists all 6-cycles of the $n$-node input graph $G$ in $O\bk[\big]{(n^2 + t) \log n}$ time, where $t$ denotes the total number of 6-cycles in $G$.
\end{theorem}
We will show how \cref{thm:all} implies our main \cref{thm:mainlist} (which has a given parameter $t$ possibly much smaller than the total cycle count) in \cref{sec:listt}.

To prove \cref{thm:all}, we consider a color-coded version of the problem: each node in $G$ receives one out of four possible colors, 
and thus the node set is partitioned into four color classes $V = A \sqcup B \sqcup C \sqcup D$.

Then, our task is to list all 6-cycles \TwoPaths, where $a \in A$, $b_1 \ne b_2 \in B$, $c_1 \ne c_2 \in C$, and $d \in D$. 
Solving this task is the main part of our algorithm,  as summarized by the following lemma:
\begin{lemma}
  \label{lm:subroutine}
  Given a 4-partite undirected graph $G=(V,E)$ where $V=A \sqcup B \sqcup C \sqcup D$,
  we can list all 6-cycles whose 6 nodes are in $A, B, C, D, C, B$ respectively in the order they appear on the cycle, in $O(n^2 + t)$ total time, where $t$ denotes the total number of 6-cycles in $G$.
\end{lemma}
We stress that in the statement of \cref{lm:subroutine}, $t$ is the total number of \emph{all} 6-cycles in $G$, not just those following the specified color pattern $A, B, C, D, C, B$.

Note that \cref{lm:subroutine} implies \cref{thm:all} by the standard color-coding technique \cite{Alon95ColorCoding}: if we color each node in $G$ with one of the four colors chosen independently at random, then any fixed 6-cycle in $G$ satisfies the specified color pattern (and hence will be reported by \cref{lm:subroutine}) with probability $\ge 4^{-6} = \Omega(1)$, so repeating $O(\log n)$ rounds suffices to report all $6$-cycles of $G$ with $1-1/\poly(n)$ success probability and total running time $O\bk[\big]{(n^2 + t) \log n}$. This color-coding can also be derandomized using the perfect hashing technique described in \cite[Section 4]{Alon95ColorCoding}, with the same asymptotic time complexity.

The rest of this section is devoted to the proof of \cref{lm:subroutine}.
Our algorithm runs in two stages: in the first stage (\cref{subsec:stage1}) we preprocess several tables, and in the second stage (\cref{subsec:stage2}) we report $6$-cycles based on the information in these tables.
Then \cref{subsec:bound} contains the key argument for bounding the time complexity of the algorithm.

\subsection{Stage I: Compute Tables}
\label{subsec:stage1}
We are given an input graph $G=(V,E)$ where $V=A \sqcup B \sqcup C \sqcup D$.
We use $t$ to denote the total number of 6-cycles in $G$.
By convention, we will use lowercase letters $a, b, c, d$ to denote nodes in $A, B, C, D$ respectively, if not otherwise stated.
We use $N_x:= \{u: (u,v)\in E\}$ to denote the set of neighbors of node $x\in V$.

We define the following tables:
\begin{itemize}
\item $N_{a, c}$ is the set of common neighbors $b \in B$ of $a$ and $c$. Formally, $N_{a, c} \defeq N_a \cap N_c \cap B$. Similarly, we define another table $N_{b, d} \defeq N_b \cap N_d \cap C$.
\item Let $P_{a, d}$ be the set of $b \in N_a \cap B$ such that $|N_{b,d}|\ge 2$, that is, the set of $b$ that can form the shape \PathForP. Similarly, let $Q_{a, d}$ be the set of $c \in N_d \cap C$ such that $|N_{a,c}|\ge 2$,  that is, the set of $c$  that can form \PathForQ.
\item Let $R_{a, d}$ be the set of edges $(b, c)\in E\cap ((B\cap N_a)\times (C\cap N_d))$ such that
  $|N_{a,c}|=|N_{b,d}|=1$. In other words, 
  there is a path $a - b - c - d$ without \emph{replacements} of $b$ and $c$ (there does not exist any other path $a - b' - c - d$ or $a - b - c' - d$ with $b' \ne b$ or $c' \ne c$).
\end{itemize}

Each of these tables $N,P,Q,R$ consists of $O(n^2)$ entries, where each entry is a set of nodes (or node pairs). We say the size of the table is the total number of nodes (or node pairs) in all its entries. The first stage of our algorithm is to compute all these tables in $O(n^2 + \textup{table size})$ time, shown in \cref{alg:compute_tables}. (Later in \cref{subsec:bound}  we will show that the total size of the tables is also $O(n^2 + t)$.)

\begin{algorithm}[htbp]
  \caption{Compute the tables}
  \label{alg:compute_tables}
  \DontPrintSemicolon
  
  \SetKwProg{Fn}{Function}{:}{}

  Let $N_{a, c} = N_{b, d} = \emptyset$ for all $a \in A$, $b \in B$, $c \in C$, $d \in D$\;
  Let $P_{a, d} = Q_{a, d} = R_{a, d} = \emptyset$ for all $a \in A$, $d \in D$\;
  \For(\Comment{Compute tables $N, P, Q$}){$(b, c) \in E \cap (B \times C)$} {
    \For{$a \in N_b \cap A$} {
      Insert $b$ to $N_{a, c}$ \hspace{5.4cm}\Comment{$b$ is a common neighbor of $a$ and $c$}\;\label{line:insert1}
      \If(\Comment{Found two paths $a-b-c$, $a-b'-c$ ($b'\neq b$)}){$\abs{N_{a, c}} = 2$} {
        Insert $c$ to $Q_{a, d}$ for all $d \in N_c \cap D$\;\label{line:insert2}
      }
    }
    \For(\Comment{Symmetric to the above}){$d \in N_c \cap D$} {
      Insert $c$ to $N_{b, d}$\;\label{line:insert3}
      \If{$\abs{N_{b, d}} = 2$} {
        Insert $b$ to $P_{a, d}$ for all $a \in N_b \cap A$\;\label{line:insert4}
      }
    }
  }
  \For(\Comment{Compute table $R$}){$(b, c) \in E \cap (B \times C)$} {
    Compute $S_a \defeq \BK{a \in N_b \cap A : \abs{N_{a, c}} = 1}$\;\label{line:S_a}
    Compute $S_d \defeq \BK{d \in N_c \cap D : \abs{N_{b, d}} = 1}$\;\label{line:S_b}
    \For{$(a, d) \in S_a \times S_d$} {
      Insert $(b, c)$ to $R_{a, d}$\;\label{line:insert5}
    }
  }
\end{algorithm}

It is straightforward to verify that \cref{alg:compute_tables}  correctly computes all the tables $N,P,Q,R$ according to their definitions. To analyze its time complexity, we only need to notice that every time \cref{line:insert1,line:insert2,line:insert3,line:insert4,line:insert5} are executed, there will be a new element inserted to some table entry, so the running time of these lines add up to $O(\textup{table size})$. The time consumed by \cref{line:S_a,line:S_b} is bounded by the number of paths $a-b-c$ and $b-c-d$ in the graph, which equals the total size of the tables $N_{a, c}$ and $N_{b, d}$, which is again $O(\textup{table size})$. The other lines of \cref{alg:compute_tables} take $O(\abs{E}) = O(n^2)$ time.

\newpage
\subsection{Stage II: Report 6-Cycles}
\label{subsec:stage2}
In the second stage, we use the precomputed tables $N,P,Q,R$ to help us report all 6-cycles of the form \TwoPaths, where $a \in A$, $b_1 \ne b_2 \in B$, $c_1 \ne c_2 \in C$, and $d \in D$.

We think of each element $b \in P_{a, d}$ as representing a collection of paths $a-b-(?)-d$ that share the same $b$ (where the $?$ mark can be any node in $N_{b,d}$), and similarly each $c \in Q_{a, d}$ represents a collection of paths sharing $c$, and each $(b, c) \in R_{a, d}$ represents a single path from $a$ to $d$.
Every possible path $a-b-c-d$ is contained in at least one of these three tables: if there is a replacement for $b$ or $c$, then this path appears in $Q_{a,d}$ and/or $P_{a,d}$; otherwise, this path is included in $R_{a,d}$.
For a desired 6-cycle \TwoPaths, we make a case distinction according to which tables the two paths $a - b_1 - c_1 - d$ and $a - b_2 - c_2 - d$ belong to:
\begin{enumerate}
\item One path belongs to $P$, and the other path belongs to $Q$.
\item Both paths belong to $P$ (or both paths belong to $Q$).
\item Both paths belong to $R$.
\item One path belongs to $P$ (or $Q$), while the other belongs to $R$.
\end{enumerate}
In the following, we separately consider each case and describe our algorithm for listing all the 6-cycles in that case.
(These cases are not disjoint, but this will only cause each cycle to be reported $O(1)$ times, which effectively blows up the total time complexity by a constant factor.) 

\paragraph*{Case 1.}
For every $a$ and $d$, we want to report all 6-cycles formed by pasting together two paths $a-b_1-c_1-d$ and $a-b_2-c_2-d$ that belong to $P$ and $Q$ respectively. To do this, we iterate over all $b_1 \in P_{a, d}$ and $c_2 \in Q_{a, d}$, and report all possible choices of $c_1 \in N_{b_1, d}$ and $b_2 \in N_{a, c_2}$ such that $c_1\neq c_2$ and $b_1 \neq b_2$.  The pseudocode is given in \cref{alg:report_case1}. It is clear that \cref{alg:report_case1} correctly lists all desired 6-cycles of Case 1.

\begin{algorithm}[htbp]
  \caption{Report 6-Cycles of Case 1}
  \label{alg:report_case1}
  \DontPrintSemicolon

  \For{each $a \in A$, $d \in D$ such that $|P_{a,d}|\ge 1$ and $|Q_{a,d}|\ge 1$} {\label{line:case1/loop1}
    \For{each $b_1 \in P_{a, d}$ and $c_2 \in Q_{a, d}$} {\label{line:case1/loop2}
      \For(\Comment{List cycles given $a, d, b_1, c_2$}){each $c_1 \in N_{b_1, d}$ and $b_2 \in N_{a, c_2}$} {\label{line:case1/inner_loop_start}
        \If{$b_1 \ne b_2$ and $c_1 \ne c_2$} {
          Report cycle \TwoPaths\;\label{line:case1/inner_loop_end}
        }
      }
    }
  }
\end{algorithm}

Next, we analyze the time complexity of \cref{alg:report_case1}. First, the number of iterations of the for-loop on \cref{line:case1/loop1} is bounded by $O(n^2)$. We then show that, for every $O(1)$ time we spend on the for-loop on \cref{line:case1/inner_loop_start} to~\ref{line:case1/inner_loop_end}, we can report one 6-cycle, thus the time complexity for this part is $O(t)$.
At every time we enter this inner for-loop on \cref{line:case1/inner_loop_start}, the time we spend equals $O(|N_{b_1, d}| \cdot |N_{a, c_2}|)$, and the number of 6-cycles we report equals the number of $(c_1, b_2)\in N_{b_1,d}\times N_{a,c_2}$ such that $b_1 \ne b_2$, $c_1 \ne c_2$, which is 
\[
  \abs*{\bk[\big]{N_{b_1, d} \setminus \BK{c_2}} \times \bk[\big]{N_{a, c_2} \setminus \BK{b_1}}} \ge (|N_{b_1, d}| - 1)(|N_{a, c_2}| - 1) \ge \frac{1}{4} |N_{b_1, d}| \cdot |N_{a, c_2}|,
\]
where the last inequality uses the fact that $|N_{b_1, d}|, |N_{a, c_2}| \ge 2$ due to the definition of $P, Q$. This implies the desired time complexity $O(n^2 +  t)$.

\paragraph*{Case 2.}
The algorithm for Case 2 is similar, see \cref{alg:report_case2}. (We only give the pseudocode for cycles whose both paths belong to $P$; the case for $Q$ is symmetric.)

\begin{algorithm}[htbp]
  \caption{Report 6-Cycles of Case 2}
  \label{alg:report_case2}
  \DontPrintSemicolon

  \For{each $a \in A$, $d \in D$ such that $|P_{a,d}|\ge 2$} {\label{line:case2/loop1}
    \For{each $b_1, b_2 \in P_{a, d}$ where $b_1 \ne b_2$} {\label{line:case2/loop2}
      \For(\Comment{List cycles given $a, d, b_1, b_2$}){each $c_1 \in N_{b_1, d}$ and $c_2 \in N_{b_2, d}$} {\label{line:case2/inner_loop_start}
        \If{$c_1 \ne c_2$} {
          Report cycle \TwoPaths\;\label{line:case2/inner_loop_end}
        }
      }
    }
  }
\end{algorithm}

Similar to Case 1, the correctness is clear, and we only need to show that every $O(1)$ time spent on \cref{line:case2/inner_loop_start} to~\ref{line:case2/inner_loop_end} will report a 6-cycle. At every time we enter the inner for-loop, we spend $O(|N_{b_1, d}| \cdot |N_{b_2, d}|)$ time, and the number of reported cycles equals the number of $(c_1, c_2)\in N_{b_1,d}\times N_{b_2,d}$ such that $c_1 \ne c_2$, which is
\[
  |N_{b_1, d}| \cdot |N_{b_2, d}| - |N_{b_1, d} \cap N_{b_2, d}| \ge |N_{b_1, d}| \cdot (|N_{b_2, d}| - 1) \ge \frac{1}{2} |N_{b_1, d}| \cdot |N_{b_2, d}|,
\]
where we again used the fact that $|N_{b_2,d}|\ge 2$ due to the definition of $P$.
Therefore, the time consumption on the inner loop is $O(t)$, and the total time complexity is the same as Case 1, $O(n^2 +  t)$.

\paragraph*{Case 3.}  For Case 3, where both paths belong to $R$, we use a straightforward enumeration as shown in \cref{alg:report_case3}.

\begin{algorithm}[htbp]
  \caption{Report 6-Cycles of Case 3}
  \label{alg:report_case3}
  \DontPrintSemicolon

  \For{each $a \in A$, $d \in D$ such that $|R_{a,d}|\ge 2$} {\label{line:case3/loop1}
    \For{each $(b_1, c_1), (b_2,c_2) \in R_{a, d}$ where $(b_1, c_1)\neq (b_2,c_2)$} {\label{line:case3/loop2}
      Report cycle \TwoPaths\;\label{line:case3/inner_loop_end}
    }
  }
\end{algorithm}
We need to show that the cycles reported by \cref{alg:report_case3} are valid.
Suppose to the contrary that $b_1=b_2$, then $c_1\neq c_2$, and hence $|N_{b_1,d}|\ge 2$. Since $(b_1,c_1)\in R_{a,d}$, this contradicts the definition of $R_{a,d}$. Hence, we must have $b_1\neq b_2$, and similarly $c_1\neq c_2$, so the reported 6-cycle is valid. Clearly, the running time of \cref{alg:report_case3} is also $O(n^2+t)$.

\paragraph*{Case 4.} Case 4 is similar to Case 3 where we use the fact that paths in $R$ have no replacements for nodes $b$ and $c$, thus $b_1 = b_2$ or $c_1 = c_2$ will never happen in the straightforward enumeration. We omit the pseudocode. The time complexity of this case is $O(n^2+t)$ as before.

\bigskip
To summarize, we have shown that our stage 2 algorithm lists all desired 6-cycles in $O(n^2+t)$ total time.
\subsection{Bounding the Total Table Size}
\label{subsec:bound}
We have shown that stage 1 runs in $O(n^2 + \textup{table size})$ time and stage 2 runs in $O(n^2 + t)$ time.
To bound the total running time, it remains to show that the total size of the tables $P,Q,R,N$ does not exceed $O(n^2 + t)$.

First, we argue that the sizes of tables $P, Q, R$ are $O(n^2 + t)$.
For any fixed $a,d$, when the set $P_{a, d}$ contains $x \ge 2$ elements, it produces at least $\binom{x}{2}\ge x - 1$ valid 6-cycles as we have seen in Case 2 of the proof in \cref{subsec:stage2}. So $t\ge \sum_{a,d}(|P_{a,d}|-1) \ge (\sum_{a,d}|P_{a,d}|) -n^2$.
A similar argument applies to tables $Q$ (Case 2 in \cref{subsec:stage2}) and $R$ (Case 3 in \cref{subsec:stage2}) as well. Hence, the total size of $P,Q,R$ is $O(n^2+t)$.

It remains to analyze the size of table $N$. Our analysis on the sizes of $N_{a, c}$ (and symmetrically, $N_{b, d}$) relies on the following lemma:
\begin{lemma}
  \label{clm:saturation}
  If node $a \in A$ satisfies $\sum_{c \in C} |N_{a, c}| \ge 100n + k$, then $G$ contains at least $k$ 6-cycles in which $a$ is the only node in $A$.
\end{lemma}

\begin{proof}\footnote{We thank an anonymous reviewer for suggesting this simpler proof; our original proof had to invoke a lemma from \cite{Yuster97FindEven} and \cite{bondy}.}
  For any fixed $a\in A$, we consider the subgraph $G_a$ of $G$ defined by taking the union of the edges in all paths $a - b - c$ where $c \in C$ and $b \in N_{a, c}$. Since these paths have distinct $(b,c)$ edges,  
  we know $|E(G_a)| \ge \sum_{c \in C} |N_{a, c}| \ge 100n + k$.
  In the following, we first prove that there exists at least one 6-cycle provided $|E(G_a)| \ge 100n$.

  We iteratively remove nodes in $G_a$ whose current degree is less than 3, and denote the resulting graph by $G'_a$, with number of edges $|E(G'_a)| \ge |E(G_a)| - 2n\ge 98n$.
  As shown in \cref{fig:graph}, $G'_a$ is composed of a bipartite graph between $V(G'_a)\cap B$ and $V(G'_a) \cap C$, and edges connecting $a$ with \emph{every} node in $V(G'_a) \cap B$.
   We know $a$ itself is not removed, since otherwise $|V(G'_a) \cap B| \le 2$ and the graph $G'_a$ can only have at most $\sum_{b\in V(G'_a) \cap B}\deg(b) \le 2n<98n$ edges. 
   
   Next, we find a 6-cycle in $G'_a$ of the form $a - b_1 - c_1 - b_2 - c_2 - b_3 - a$ in a greedy fashion: starting from an arbitrary $b_1\in V(G'_a)\cap B$, in each step we go to an arbitrary unvisited neighbor of the current node (which must exist since every node has degree $\ge 3$ in $G'_a$) as the next node on the cycle.
  Hence, we have found a 6-cycle containing $a$ in $G_a$.

  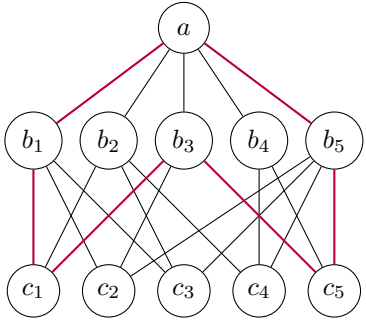
\begin{figure}[htbp]
    \centering
    \begin{tikzpicture}
      \foreach \i in {1,...,5} {
        \ifnum\i=1
          \node[draw, circle] (b\i) at (\i,1) {$b_\i$};
        \else
          \node[draw, circle] (b\i) at (\i,1) {$b_\i$};
        \fi
      }

      \node[draw, circle] (a) at (3,2.5) {$a\vphantom{b_i}$};

      \foreach \i in {1,...,5} {
        \draw (a) -- (b\i);
      }

      \foreach \i in {1,...,5} {
        \node[draw, circle] (c\i) at (\i,-1) {$c_\i$};
      }

      \draw (b1) -- (c1);
      \draw (b2) -- (c3);
      \draw (b3) -- (c5);
      \draw (b4) -- (c4);
      \draw (b5) -- (c2);
      \draw (b1) -- (c2);
      \draw (b2) -- (c4);
      \draw (b3) -- (c1);
      \draw (b4) -- (c5);
      \draw (b5) -- (c3);

      \draw (c1) -- (b2);
      \draw (c2) -- (b3);
      \draw (c3) -- (b1);
      \draw (c4) -- (b5);

      \draw[thick, purple] (a) -- (b1) -- (c1) -- (b3) -- (c5) -- (b5) -- (a);
    \end{tikzpicture}
    \caption{An example of $G'_a$. We are guaranteed that every node including $a$ has degree at least 3, so we can find a 6-cycle step-by-step (drawn in thick purple lines).}
    \label{fig:graph}
  \end{figure}

  Take the edge $(b_1,c_1)$ from the 6-cycle that we have just found, and remove $(b_1,c_1)$ from $G_a$, so that this 6-cycle no longer appears in the new $G_a$. 
  Then, the number of edges $|E(G_a)|$ decreases by 1; if it is still $\ge 100n$, we can repeat the above argument to find a second 6-cycle.
  This process is repeated for $k$ times as the initial $G_a$ has $100n + k$ edges, and $k$ cycles are reported in total, which proves the lemma.
\end{proof}

According to \cref{clm:saturation}, for any fixed $a\in A$, $\sum_{c \in C} |N_{a, c}| - 100n$ is a lower bound on the number of 6-cycles in $G$ in which $a$ is the only node in $A$. Thus, $\sum_{a \in A} \sum_{c \in C} |N_{a, c}| - 100n^2$ is a lower bound on the total number of 6-cycles in $G$. In other words, the total size of $N_{a, c}$ over all $a,c$ is $O(n^2 + t)$. 
A symmetric argument shows that the total size of $N_{b,d}$ over all $b,d$ is also $O(n^2+t)$.
Hence, the total size of all tables $N, P, Q, R$ is bounded by $O(n^2 + t)$. 
So our algorithm in stage 1 also takes $O(n^2+\textup{table size})\le O(n^2 + t)$ time.
This finishes the proof of \cref{lm:subroutine}.

\section{Algorithm for Listing $t$ 6-Cycles}
\label{sec:listt}

It remains to show how our algorithm for listing all 6-cycles (\cref{thm:all}) can be modified into an algorithm for listing $t$ 6-cycles in $\tilde O(n^2+t)$ time, for any given $t$.\footnote{A natural attempt is to subsample the edges of $G$ at some rate $p$ and obtain a subgraph $G'$ whose expected $C_6$ count (which equals $p^6$ times the original count) is close to $t$, and then apply \cref{thm:all} on $G'$. This does not directly work since the $C_6$ count of $G'$ can have high variance (for example, when some edge participates in many $C_6$s in $G$) and can be far from $t$ with very high probability.  Our approach here instead uses a binary search and is deterministic.}

Let the nodes of $G=(V,E)$ be $v_1,v_2,\dots,v_n$ in arbitrary order. Let $G_i$ be the subgraph of $G$ induced by the first $i$ nodes $v_1,\dots,v_i$.
Note that we can decide whether the number of 6-cycles in $G_i$ is at least $t$, in $\tilde O(|V(G_i)|^2 + t)\le \tilde O(n^2 + t)$ time: we simply attempt  to list all 6-cycles in $G_i$ using the algorithm from \cref{thm:all}, and we terminate the algorithm if the running time is too long. This allows us to use binary search to find the largest index $1\le i\le n$ such that the number of 6-cycles in $G_i$ is at most $t$, in $\tilde O(n^2 + t)$ time. Now there are three cases:
\begin{itemize}
\item Case 1: $i=n$. 

  This means the total number of 6-cycles in $G$ is at most $t$, and can be listed by \cref{thm:all} in $\tilde O(n^2+t)$ time.
\item Case 2: $i<n$.

  This means $G_{i+1}$ contains more than $t$ 6-cycles.
  We then decide whether $G_{i+1}$ has at most $2t$ 6-cycles, in $\tilde O(n^2+t)$ time. Based on the outcome there are two cases:
  \begin{itemize}
  \item Case 2(1):   $G_{i+1}$ has at most $2t$ 6-cycles.

    In this case we can list all  6-cycles in $G_{i+1}$ in $\tilde O(n^2 + t)$ time, and output $t$ of them as the answer.
  \item Case 2(2): $G_{i+1}$ has more than $2t$ 6-cycles.

    Since $G_i$ has at most $t$ 6-cycles, we know $v_{i+1}$ is contained in more than $2t-t=t$ 6-cycles.

    In this case, we simply apply the following standard lemma to $v_{i+1}$, and output $t$ 6-cycles containing $v_{i+1}$ in $\tilde O(m+t) = \tilde O(n^2+t)$ time. 
  \end{itemize}
\end{itemize}
\begin{lemma}[Simple adaptation of {\cite[Lemma 3.1]{Alon95ColorCoding}}]
  For fixed $k$, given an $n$-node $m$-edge graph $G$ with a special node $v$, and a parameter $t$, we can list $t$ $k$-cycles in $G$ containing $v$ in $O((m+t)\log n)$ time.
\end{lemma}
\begin{proof}[Proof Sketch]
  Using the color coding technique \cite{Alon95ColorCoding}, we can assume the nodes in $G$ are partitioned as $V(G) = \{v\}\sqcup V_1\sqcup V_2 \sqcup \dots \sqcup V_{k-1}$, and we only need to list $k$-cycles with nodes $v,v_1\in V_1,\dots,v_{k-1}\in V_{k-1}$ in order. By a dynamic-programming-like procedure in $O(m)$ time, we can compute all the vertices $v_i\in V_i$ that are reachable from $v$, as well as all their parents $v_{i-1}\in V_{i-1}$ reachable from $v$ and adjacent to $v_i$. To report $k$-cycles, we start with every neighbor of $v$ in $V_{k-1}$, and follow the parent pointers to go back to $v$ using DFS. In this way we output $t$ $k$-cycles in $O(kt) = O(t)$ time.
\end{proof}

\bibliographystyle{alpha}
\bibliography{reference}

\end{document}